\documentclass[copyright,creativecommons]{eptcs}
\providecommand{\event}{F. Bonchi, D. Grohmann, P. Spoletini, and E. Tuosto:\\ ICE'09 Structured Interactions} % Name of the event you are submitting to
\providecommand{\volume}{nnn}
\providecommand{\anno}{2009}
\providecommand{\firstpage}{57}
\providecommand{\eid}{4}

\usepackage[latin9]{inputenc}
\usepackage{amsmath,amssymb,amsthm}
\usepackage{graphicx}
\usepackage{url}
\usepackage{tabularx}
\usepackage{subfigure}
\usepackage{array}
\usepackage{pgf}
\usepackage{tikz}
\usepackage[all]{xy}

\newtheorem{definition}{Definition}
\newtheorem{theorem}{Theorem}
\newtheorem{remark}{Remark}
\newtheorem{example}{Example}

\newcommand{\Set}{\ensuremath{\mathbf{Set}}}
\newcommand{\PA}{\ensuremath{\mathbf{PortAut}}}
\newcommand{\Connector}{\ensuremath{\mathbf{Connector}}}
\newcommand{\one}{\ensuremath{\mathbf{1}}}
\newcommand{\Sem}{\ensuremath{\mathcal{S}em}}
\newcommand{\N}{\ensuremath{\mathcal{N}}}
\newcommand{\A}{\ensuremath{\mathcal{A}}}
\newcommand{\C}{\ensuremath{\mathcal{C}}}
\newcommand{\Q}{\ensuremath{\mathcal{Q}}}

\usetikzlibrary{arrows,shapes,automata,backgrounds,fit,automata,petri}

\definecolor{cblue}{rgb}{0.87,0.91,0.95}
\tikzstyle{cbox} = [fill=cblue,draw=gray,thick,rounded corners,
                    inner sep=3mm,outer xsep=1.5mm] 
\tikzstyle{every state}=[node distance=2.0cm,inner sep=0pt,minimum
size=6mm,draw=black,fill=white,initial text=,font=\footnotesize]

\tikzstyle{place}=[circle,thick,draw=black,fill=white,minimum size=3mm]
\tikzstyle{transition}=[rectangle,thick,draw=black,
  			  fill=black!20,minimum size=3mm]
\tikzstyle{every token}=[minimum size=1.1mm]
\tikzstyle{pre}=[thick,bend angle=32,<-]
\tikzstyle{post}=[thick,bend angle=32,->]

\newcommand{\tl}[1]{\scriptsize{\ensuremath{\mathsf{#1}}}}

\title{Integrated Structure and Semantics for Reo~Connectors~and~Petri~Nets}
\author{Christian Krause
	\footnote{Supported by NWO projects WoMaLaPaDiA and SYANCO.}
	\email{c.krause@cwi.nl}
	\institute{
  		CWI, P.O. Box 94079, 1090GB Amsterdam, The Netherlands
	}
}

\def\titlerunning{Integrated Structure and Semantics for Reo~Connectors~and~Petri~Nets}
\def\authorrunning{Christian Krause}

\begin{document}

\maketitle

\begin{abstract}
In this paper, we present an integrated structural and behavioral
model of Reo connectors and Petri nets, allowing a direct comparison 
of the two concurrency models. 
For this purpose, we introduce 
a notion of \emph{connectors}
which consist of a number of interconnected, 
user-defined primitives with fixed behavior.
While the structure of connectors resembles
hypergraphs, their semantics is given
in terms of so-called
\emph{port automata}. We define both
models in a categorical setting where composition
operations can be elegantly defined and
integrated. Specifically, we formalize structural gluings of
connectors as pushouts, and joins of port automata as
pullbacks. We then define a semantical functor from the connector to
the port automata category which
preserves this composition. We further show how to encode
Reo connectors and Petri nets into this model and indicate
applications to dynamic reconfigurations modeled using 
double pushout graph transformation.
\end{abstract}

\section{Introduction}

Reo~\cite{Arbab04} is a channel-based coordination language which
has its main application area in component and service composition.
The idea in Reo is to construct complex, so-called \emph{connectors}
out of a set of user-defined primitives, most commonly channels.
Among a number of sophisticated features, such as mobility~\cite{SABB06}, 
context-dependency~\cite{CCA07,BCS09}\nocite{Tiles}\nocite{Clarke07} and dynamic
reconfigurability~\cite{KMLA10}, on a more basic level Reo can be 
seen also as a model of concurrency. 
Comparing Reo with Petri nets, the first obvious commonality is the fact that
they both use a graph-based model, i.e. their structure can be modeled using 
typed graphs. Moreover, both models combine control-flow and data-flow
aspects. In this paper, we are particularly interested in
the concurrency properties of the two models, i.e.
parallel or synchronized actions vs. interleaved or mutually
excluded actions. To understand the relationship between
Reo connectors and Petri nets, we follow an approach in
this paper where we map both models to so-called \emph{port automata}~\cite{KC09},
which serve as our common semantical domain. We can thereby gain an integrated
view on structure and semantics of Reo connectors and Petri nets and moreover 
compare both models.

As a motivating example, Fig.~\ref{fig:examples} depicts a Reo connector,
a Petri net and a port automaton, all modeling the same simple protocol.
If considering the initial state also as final, the accepted language
is $(AB+CD)^{*}$. Port automata model explicitly synchronization of actions.
This is witnessed by the fact that the transitions in the automaton are 
\emph{sets} of truly concurrent actions. Such a port automaton transition
corresponds to a concurrent firing of transitions in a Petri net, or a synchronized activity
on nodes in a Reo connector. This is our starting point for
using port automata as a common semantical models for the structural
models of Reo connectors and Petri nets. Our general idea is to compose 
-- potentially user-defined -- primitives into a graph-structure which we will
refer to as \emph{connector}. While in Reo, these primitives are communication channels,
in Petri nets we consider places as primitives.
Moreover our approach emphasizes compositionality, i.e.
the port automata semantics of primitives is predefined, but the semantics 
of connectors is derived using a join-operation.

\begin{figure}[t]
\centering
\input{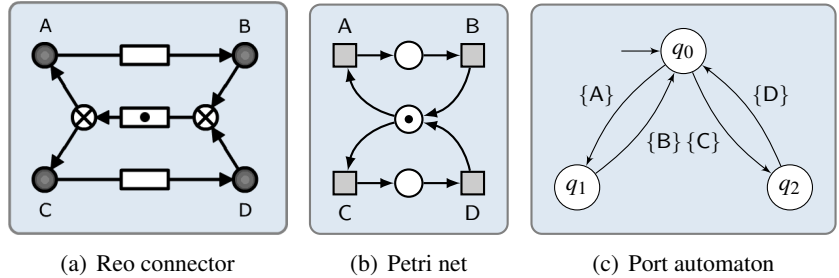}
\caption{Three descriptions of the same protocol.}
\label{fig:examples}
\end{figure}

We make the following contributions in this paper. We present
a model of connectors which combines 
structural and behavioral descriptions.
The structure of connectors resembles hypergraphs and their
semantics is defined using port automata. 
Our most important result is the compositionality of the model 
in the following sense: a structural gluing of connectors 
corresponds to a join of the corresponding port automata.
For this purpose, we define the categories {\Connector} and 
{\PA}, and a contravariant functor $\Sem: \Connector \to \PA^{op}$.
In categorical terms, our compositionality result means that this
functor sends pushouts of connectors to pullbacks of the 
corresponding port automata, i.e. for connectors $\C_0$, $\C_1$ and
$\C_2$:
\[
\Sem(\C_1 +_{\C_0} \C_2) = \Sem(\C_1) \times_{\Sem(\C_0)} \Sem(\C_2).
\]
Furthermore we show how Reo connectors and Petri nets can be modeled
directly in our framework. While for Petri nets, compositionality results 
similar to ours exist already, this paper constitutes the first 
formal integration of the graph structure and the automata semantics 
of Reo connectors. Further it is a starting point for synthesis 
algorithms and in particular for semantics of graph transformation 
based reconfigurations.  Specifically, our composition operation fits into the 
double pushout approach~\cite{DPO,EEPT06} for graph transformation,
which has been used for instance to model reconfigurations 
of Reo connectors in~\cite{KMLA10}, and of Petri nets 
in~\cite{PER01}.

\paragraph{Organization.} The rest of the paper is organized
as follows. We start with the semantical model by introducing
port automata in Section~\ref{sec:port-automata}. Based on this,
we then define our notion of connectors in Section~\ref{sec:connectors}.
Section~\ref{sec:semantics} contains our main compositionality
result and Section~\ref{sec:encodings} shows how Reo connectors
and Petri nets can be encoded in our connector model.
Finally, Section~\ref{sec:discussion} contains a discussion and 
future work, and Section~\ref{sec:related-work} includes related work.

%\newpage
\vspace{0.5cm}
\section{Port automata}
\label{sec:port-automata}

Port automata are an operational model for connectors
and have been mainly studied in the context of Reo.
They are an abstraction of so-called \emph{constraint automata}~\cite{BSAR06}
which is the quasi-standard semantics of Reo.
Port automata describe the synchronization on sets
of ports, depending on the internal state of the connector.
The model abstracts from both the direction and content
of data flow. For a proper modeling of data we refer to the
constraint automata model.

In this paper, we present port automata
in a categorical setting, i.e., we consider
them as objects in a category which we will denote with~\PA.
We now give the definition for port automata.
\vspace{0.2cm}
\begin{definition}[Port automaton] 
A port automaton $A =(Q,N,\to,i)$ consists of a set of states $Q$, 
a set of port names $N$, a transition relation $\to \subseteq Q \times 2^{N}
\times Q$ and an initial state $i \in Q$.
\end{definition}
We denote transitions often as $q\overset{S}{\longrightarrow}p$
with $q,p \in Q$ and $S\subseteq N$. 
The interpretation is that there is concurrent activity
at the ports $S$ and no activity at the rest of the ports $N\backslash S$. 
The model permits $\tau$-transitions, namely whenever $S=\emptyset$. Hence,
there can be silent steps without any action. In the following we define
a notion of port automata morphism.
\vspace{0.2cm}
\begin{definition}[Port automata morphism]
A morphism of port automata $f:A_1 \to A_2$ is a pair of functions $f=(f_Q,f_N)$
with $f_Q: Q_1 \to Q_2$ and $f_N: N_2 \to N_1$, such that: 
%\begin{itemize}
%\item 
$f_Q(i_1)=i_2$ and 
%\item 
for all transitions $q \overset{S_1}{\longrightarrow}_1 p$ in $A_1$
there exists a transition $f_Q(q) \overset{S_2}{\longrightarrow}_2
f_Q(p)$ in $A_2$ with
\begin{equation}
	\label{equ:aut-mor}
	f_N(N_2) \cap S_1 = f_N(S_2).
\end{equation}
%\end{itemize}
\end{definition}

\vfill

Port automata morphisms can be seen as a kind of simulation.
%Similar to the notion of automata morphisms in~\cite{DS02}, 
The definition uses
a function for relating the states of the automata instead of a relation, which
one might expect for a simulation of automata. However, in our categorical
context, especially when mapping connector morphisms to (inverse) simulations,
this definition is sufficient and easier to handle.
Note further that the port names are mapped in the opposite direction and
that condition~(\ref{equ:aut-mor}) defines $S_2$ as the restricted preimage of
$S_1$. The following example illustrates this notion of automata morphisms.

\vfill

\begin{example}
An example of a port automata morphism is depicted in
Fig.~\ref{fig:automata-morphism}. States $q_0,q_2$ are both mapped to $p_0$, and
$q_1$ is mapped to $p_1$. The port names function is the inclusion map in the
opposite direction. The transition via $\{\mathsf{B,C}\}$ in the source corresponds
to the transition via $\{\mathsf{B}\}$, and $\{\mathsf{C}\}$ to
the $\tau$-step in the target automaton.
\end{example}

\vfill

\begin{figure}[h]
	\begin{center}
	\begin{tikzpicture}

  \begin{scope}[>=latex',->];

  \node[state,initial] (q0)                {$q_0$};
  \node[state]         (q1) [right of=q0]  {$q_1$};
  \path		(q0) edge [bend left=15]	node [above,name=A] {\tl{\{A\}}} (q1);
  
  \node[state]         (q2) [below of=A,node distance=1.5cm]  {$q_2$};
  \path     (q1) edge [bend left=15]	node [below right, inner sep=1pt]
  {\tl{\{B,C\}}} (q2) (q2) edge [bend left=15]	node [below left, inner sep=1pt]
            {\tl{\{C\}}} (q0) ;

  \node[state,initial] (p0) [right of=q1,node distance=3.25cm,yshift=-3.5mm] 
  {$p_0$}; \node[state] (p1) [right of=p0]   {$p_1$};
  \path		(p0) edge [bend left=15]	node [above,name=A'] {\tl{\{A\}}} (p1)
		         edge [loop above]      node[above,name=tau2] {\tl{ \emptyset } } ()
            (p1) edge [bend left=15]	node [below] {\tl{\{B\}}} (p0);

  \end{scope}
  
  \begin{pgfonlayer}{background}
    \node [cbox,fit=(q0) (q1) (q2) (A),minimum width=4cm,minimum
    height=2.5cm,name=source,outer xsep=1.5mm] {};
    \node [cbox,right of=source,node distance=5.25cm,minimum width=4cm,minimum
    height=2.5cm,name=target,outer xsep=1.5mm] {}; 
  \end{pgfonlayer}

  \draw[->,thick,gray] (source) -- (target);
  
\end{tikzpicture}
	\end{center}
\caption{A morphism of port automata.}
\label{fig:automata-morphism}
\end{figure}
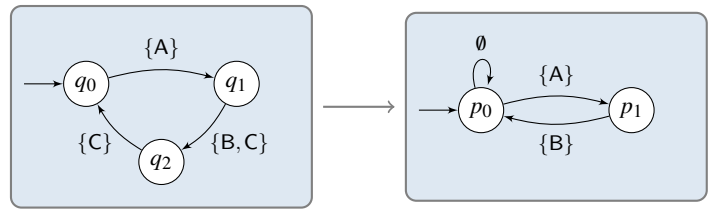

\vfill

Note that if the port name map $f_N$ is the identity, a morphism also gives
rise to a language inclusion. If there is no confusion, we abuse notation and
write $f$ for both $f_Q$ and $f_N$.
If there is a morphism between two port automata $A_0$ and $A_1$, we
may also write $A_0\succeq A_1$ for short.
Similarly, if there exists a (categorical) isomorphism, we denote
this by $A_0 \cong A_1$. Note that this notion of behavioral equivalence is
stronger than usual definitions, e.g. using bisimulations.
\vfill
Composition and identity of port automata morphisms are defined
componentwise in \Set. The resulting category of port automata is denoted by
{\PA}. The port automaton with one state, an empty port names set and
a $\tau$-transition is the final object in this category, denoted by $\one$.
At this point, we already make use of our categorical setting
and define composition of port automata using pullbacks.

\newpage
\begin{theorem}[Pullbacks of port automata]
\label{thm:pullbacks}
The category $\PA$ has pullbacks and they can be constructed componentwise
in \Set. For a cospan $A_1 \to A_0 \leftarrow A_2$, the pullback object
is $A_3 = (Q_3, N_3, \to_3, i_3)$ where\\
\begin{minipage}{0.3\textwidth}
\[\xymatrix{
  A_0 &	
  A_1 \ar[l]_{f_1} \\
  A_2 \ar[u]^{f_2}  & 
  A_3 \ar[u]^{g_1} 	\ar[l]_{g_2} \\ & & 
  X \ar@/_/[uul]_{h_1} \ar@/^/[ull]^{h_2} \ar@{.>}[ul]|-{h}
  }\]
\vspace{3mm}
\end{minipage}
\hfill
\begin{minipage}{0.65\textwidth}
\begin{itemize}
\item $Q_3 = Q_1 \times_{Q_0} Q_2$ (pullback in \Set)
\item $N_3 = N_1 +_{N_0} N_2$ (pushout in \Set)
\item $i_3 = \langle i_1,i_2\rangle$
\item if $q_1 \overset{S_1}{\longrightarrow_1} p_1$ and $q_2 \overset{S_2}{\longrightarrow_2} p_2$,
such that
\begin{equation}
	\label{equ:caut-po}
	g_1(S_1) \cap g_2(N_2) = g_2(S_2) \cap g_1(N_1) 
\end{equation}
then
	$\xymatrix{
	\langle q_1,q_2 \rangle 
	\ar[rr]^{g_1(S_1) \cup g_2(S_2)} &&
	\langle p_1,p_2 \rangle
	}$
in $A_3$.
\end{itemize}
\end{minipage}
\end{theorem}
\begin{proof}[Proof sketch]
It is sufficient to show that the componentwise construction of $g_1,g_2$ and
$h$ yields valid morphisms of port automata, i.e. that
condition (\ref{equ:aut-mor}) holds. A detailed proof is
given in the appendix.
\end{proof}

\begin{example}
An example of a port automata pullback is depicted in
Fig.~\ref{fig:port-automata-pullback}. The state maps are
indicated by the indices, e.g. $p_0$ is mapped to $q_0$
and $p_1, p'_1, p_2$ are all mapped to $q_{12}$.
The resulting automaton on the bottom right is the automaton from
our previous example in Fig.~\ref{fig:examples}. Note that
it actually includes more states which are not shown 
here because they are unreachable.
\end{example}

\begin{figure}[t]
	\begin{center}
	\begin{tikzpicture}
  
  \begin{scope}[>=latex',->];

  \node[state,initial] (q2)                {$p_0$};
  \node[state]         (q3) [below of=q2]  {$p_2$};
  \node[state]         (q4) [right of=q2]  {$p_1$};
  \node[state]         (q5) [right of=q3]  {$p'_0$};
  \path		(q2) edge[bend right=15] node [left,name=c] {\tl{\{A\}}} (q3)
  			(q3) edge[bend right=15] node [right] {\tl{\{B\}}} (q2)
          	(q2) edge[bend left=15] node [above,name=a] {\tl{\{C\}}} (q4) 
          	(q4) edge[bend left=15] node [below] {\tl{\{D\}}} (q2) 
          	(q3) edge[bend left=15] node [above] {\tl{\{D\}}} (q5) 
          	(q5) edge[bend left=15] node [below] {\tl{\{C\}}} (q3) 
          	 ;

  \node[state]         (q1) [left of=c,node distance=2.5cm]  {$q_{12}$};
  \node[state,initial] (q0) [left of=q1]   {$q_0$};
  \path	(q0) edge [bend left=60] node [above] {\tl{ \{A\} }} (q1)
  		(q0) edge [bend left=15] node [above] {\tl{ \{C\} }} (q1)
        (q1) edge [bend left=15] node [below] {\tl{ \{D\} }} (q0)
        (q1) edge [bend left=60] node [below] {\tl{ \{B\} }} (q0);

  \node[state,initial] (r0) [below of=q0, node distance=3.5cm]  {$r_0$};
  \node[state]         (r1) [right of=r0]  {$r_1$};
  \node[state]         (r2) [below of=r0]  {$r_2$};
  \node[state]         (r3) [below of=r1]  {$r'_0$};
  \path		(r0) edge[bend left=15] node [above] {\tl{\{C\}}} (r1)
  			(r1) edge[bend left=15] node [below] {\tl{\{D\}}} (r0)
          	(r0) edge[bend right=15] node [left] {\tl{\{A\}}} (r2) 
          	(r2) edge[bend right=15] node [right] {\tl{\{B\}}} (r0) 
          	(r1) edge[bend right=15] node [left] {\tl{\{B\}}} (r3) 
          	(r3) edge[bend right=15] node [right,name=c2] {\tl{\{A\}}} (r1) 
          	 ;

  \node[state,initial] (s0) [below of=q2, node distance=4.5cm]  {$s_0$};
  \node[state]         (s2) [below of=s0]  {$s_2$};
  \node[state]         (s3) [right of=s0]  {$s_1$};
  \path (s0) edge[bend right=15] node [left] {\tl{\{A\}}} (s2) 
        (s2) edge[bend right=15] node [right] {\tl{\{B\}}} (s0) 
        (s0) edge[bend left=15] node [above] {\tl{\{C\}}} (s3) 
        (s3) edge[bend left=15] node [below] {\tl{\{D\}}} (s0) ;

  \end{scope}
  
  \begin{pgfonlayer}{background}
    \node [cbox,fit=(q0) (q1),minimum width=4cm,minimum
           height=3.6cm,name=A1,outer sep=1.5mm] {}; 
    \node [cbox,fit=(q2) (q5),minimum width=4cm,minimum
           height=3.6cm,name=A2,outer sep=1.5mm] {};
    \node [cbox,fit=(r0) (r3),minimum width=4cm,minimum
           height=3.6cm,name=A3,outer sep=1.5mm] {};
    \node [cbox,right of=A3,node distance=5.0cm, minimum width=4cm,minimum
           height=3.6cm,name=A4,outer sep=1.5mm] {};
  \end{pgfonlayer}

  \draw[->,thick,gray] (A2) -- (A1);
  \draw[->,thick,gray] (A3) -- (A1);
  \draw[->,thick,gray] (A4) -- (A2);
  \draw[->,thick,gray] (A4) -- (A3);
  
\end{tikzpicture}
	\end{center}
\caption{A pullback of port automata.}
\label{fig:port-automata-pullback}
\vspace{-2mm}
\end{figure}
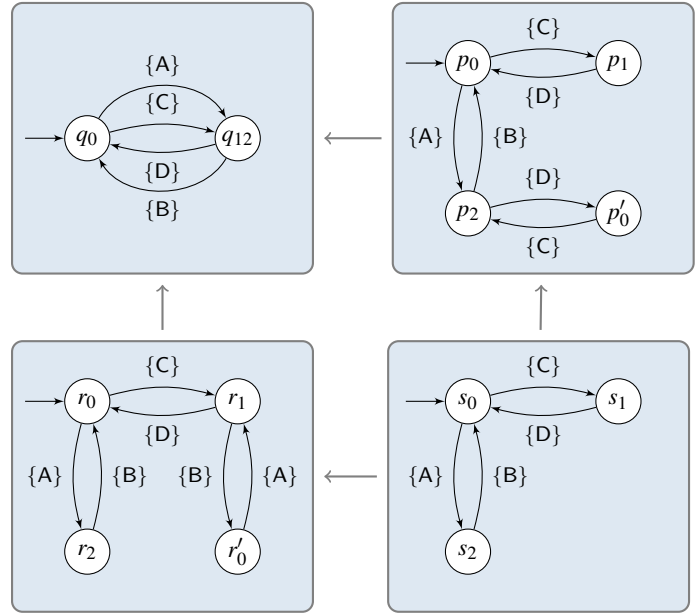

We use the default notation for pullbacks of port automata, i.e.
$A_3 = A_1 \times_{A_0} A_2$.
This notion of composition generalizes the join-operation
in~\cite{KC09} for port automata and in~\cite{BSAR06} for constraint automata
since it allows a composition along a common interface automaton. In
the traditional approaches, automata are joined only along a common set of port
names. Moreover, the categorical construction using pullbacks includes the
morphisms into the original automata and thereby relates them with the
result using simulations. Note also, that we have indirectly shown that 
{\PA} has general limits, since it has pullbacks and a final object.

In the following theorem, we phrase a basic compatibility result for port
automata morphisms, which is a direct consequence of the pullback construction.
\begin{theorem}[Compatibility with simulations]
Given two cospans of simulations $A_1 \overset{f_1}{\longrightarrow} A_0
\overset{f_2}{\longleftarrow} A_2$ and 
$B_1 \overset{g_1}{\longrightarrow} B_0 \overset{g_2}{\longleftarrow} B_2$,
then
\begin{enumerate}
  \item for three morphisms $h_i: A_i \to B_i$ with $i\in \{0,1,2\}$:
  \vspace{2mm}
  \begin{itemize} 
    \item if $h_0 \circ f_1 = h_1 \circ g_1$ and $h_0 \circ f_2 = h_2 \circ
    g_2$ then $\left( A_1 \times_{A_0} A_2 \right) \succeq \left( B_1
    \times_{B_0} B_2 \right)$.
  \end{itemize} 
  \vspace{2mm}
  \item if $A_1 \succeq B_1$ and $A_2 \succeq B_2$ then 
  $\left( A_1 \times_{A_0} A_2 \right) \succeq \left( B_1 \times B_2
  \right)$.
\end{enumerate}
\end{theorem}
\begin{proof}
Consider the following diagram where $(1)$ and $(2)$ are pullbacks of the given
cospans:
\vspace{-3mm}
\[
\xymatrix @C2mm @R2mm {
  && 
  A_0 \ar[rrrrrr]^{h_0} 
      &&&&&& 
  B_0 \\
      &&&&
  A_2 \ar[rrrrrr]^(0.25){h_2} 
      \ar[ull]_(.35){f_2} 
      \ar@{}[dllll]|{(1)} 
      &&&&&& 
  B_2 \ar[ull]_{g_2} 
      \ar@{}[dllll]|{(2)} 
      \\ 
  A_1 \ar'[rrr][rrrrrr]^(.5){h_1} 
      \ar[uurr]^{f_1} 
      &&&&&& 
  B_1 \ar'[ur][uurr] ^(.25){g_1} 
      \\ 
      && 
  A_3 \ar[uurr]_(.4){f_{2}'} 
      \ar@{.>}[rrrrrr]^{h_3} 
      \ar[ull]^{f_{1}'} 
      &&&&&& 
  B_3 \ar[ull] 
      \ar[uurr] 
      \\
}
\]
The precondition of \emph{(a)} states that the top-face and the back-face 
commute and $(1)$ commutes as it is a pullback. Hence $g_1 \circ h_1 \circ
f'_1 = g_2 \circ h_2 \circ f'_2 : A_3 \to B_0$. The morphism $h_3$ is then
uniquely determined by the pullback $(2)$ and hence $A_3 \succeq B_3$. 
For $(b)$ we take $B_0=\one$ the final object and automatically obtain
the precondition of $(a)$. Thereby:
$\left( A_1 \times_{A_0} A_2 \right) \succeq \left( B_1 \times_{\one} B_2
\right) = \left( B_1 \times B_2 \right)$.
\end{proof}

Based on the given semantical model, we are now able to enrich it
with structural aspects. We do so by introducing our notion of 
connectors in the following section.

\section{Connectors}
\label{sec:connectors}

The model that we use here is motivated by the 
idea of constructing complex connectors out of 
a set of primitives with predefined behavior.
In our context, the primitives are specified
as port automata and connectors are just
collections of port automata with overlapping
port names.
\begin{definition}[Connector] \label{def:connector}
A connector $\C = (\A,\N) $ consists of a set of port 
automata $\A$ and a set of nodes $\N$, such that $N \subseteq \N$ for all
$A=(Q,N,\to,i)\in \A$.
\end{definition}
Port names can now be interpreted as nodes and the port automata
as edges in a hypergraph. We will refer to the port automata
in a connector as \emph{primitives}. As mentioned already,
the idea is to construct arbitrarily complex connectors 
out of a fixed class of primitives, e.g. the set of 
standard channels in Reo.
\begin{definition}[Connector morphism] \label{def:connector-morphism}
A connector morphism $f: \C_1 \to \C_2$ is a pair of functions
$f = (f_\A: \A_1 \to \A_2, f_\N : \N_1 \to \N_2)$
such that for all $A=(Q,N,\to,i) \in \A_1$ there exists a port automata
morphism $f_A: f_\A(A) \to A$ with $f_A(N)=f_\N(N)$.
\end{definition}
A connector morphism consists of a map of nodes and a map of primitives 
from the source to the target connector.
Moreover, for all mapped primitives there must exist simulations in the
opposite direction, and the port name map must be compatible with the
nodes map.
It is worth mentioning at this point, that the existence of 
an inverse simulation has the consequence that a primitive can 
in principle be mapped to primitive with potentially
different interface (port name sets) and behavior 
(port automaton itself). Due to this property, connector morphisms 
permit a refinement of primitives.

Composition and identity of connector morphisms are again defined 
componentwise in \Set. We denote the category of connectors and their
morphisms as {\Connector}. 
We use pushouts to compose connectors. This makes the approach particularly
interesting for applying algebraic graph transformation techniques for modeling
reconfigurations (cf.~\cite{KMLA10,KAV09}). 

%\newpage
\begin{theorem}[Pushouts of connectors]
\label{thm:connector-pushouts}
The category {\Connector} has pushouts. 
For a span of connectors $\C_1 \leftarrow
\C_0 \to \C_2$ the pushout object is given by $\C_3 = (\A_3,\N_3)$ with\\
\begin{minipage}{0.25\textwidth}
\[\xymatrix{
  \C_0 \ar[r]^{f_1} \ar[d]_{f_2} &	
  \C_1 \ar[d]_{g_1} \ar@/^/[ddr]^{h_1} \\
  \C_2 \ar[r]^{g_2} \ar@/_/[drr]_{h_2} & 
  \C_3 \ar@{.>}[dr]|-{h} \\ & & X 
  }\]
\vspace{3mm}
\end{minipage}
\hfill
\begin{minipage}{0.7\textwidth}
\vspace{3mm}
\begin{itemize}
\item $\A_3 = \A_1 +_{\A_0} \A_2$ (pushout in \Set)
\item $\N_3 = \N_1 +_{\N_0} \N_2$ (pushout in \Set)
\vspace{3mm}
\item for all $A_0 \in \A_0$, $A_1=f_1(A_0)$ and $A_2=f_2(A_0)$:
\begin{equation}
\label{equ:connector-pushout-1}
A_3 = A_1 \times_{A_0} A_2 \in \A_3\;(pullback~in~\PA)
\end{equation}
\item for all $A_j \in \A_j \backslash f_j(\A_0), \; j\in \{1,2\}:$
\begin{equation}
\label{equ:connector-pushout-2}
 A_3 = A_j \in \A_3
\end{equation}
\end{itemize}
\end{minipage}
\end{theorem}
\begin{proof}
Due to the componentwise construction in {\Set} and {\PA} again
we have to show only that the construction yields a valid connector $\C_3$
and valid connector morphisms $g_1,g_2$ and $h$. The connector $\C_3$ is valid
since 
\[
N_3 = N_1 +_{N_0} N_2 \; \subseteq \; \N_1 +_{\N_0} \N_2 = \N_3
\vspace{0.2cm}
\]
in case (\ref{equ:connector-pushout-1}) and $N_3 = N_j \subseteq \N_j =
\N_3$ in case (\ref{equ:connector-pushout-2}).
Moreover, for every $A \in \A_j$ there exists a port automata morphism
$g_A: g_j(A) \to A$. In case (\ref{equ:connector-pushout-2}) it is the identity
and in (\ref{equ:connector-pushout-1}) it is the projection of the pullback.
Since the port name maps in $g_A$ and $g_j$ are both constructed as the injections
into $N_3$ and $\N_3$ respectively, $g_A(N) = g_{j}(N)$ holds as well.
Hence, $g_1,g_2$ are valid connector morphisms. Validity of $h$
can be shown analogously.
\end{proof}
\begin{figure}[t]
	\begin{center}
	
\begin{tabular}{ m{30mm}  m{10mm}  m{30mm}}
	\includegraphics[width=30mm]{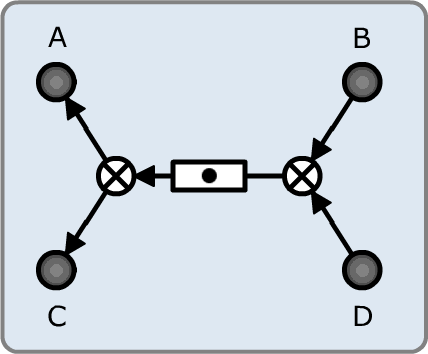} &
	\tikz{ \draw[->,thick,gray] (0,0) -- (1,0);} &
	\includegraphics[width=30mm]{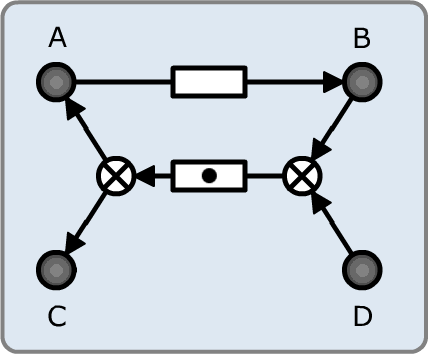} \\
	\hspace{13mm} \tikz{ \draw[->,thick,gray] (0,1) -- (0,0);} & &
	\hspace{13mm} \tikz{ \draw[->,thick,gray] (0,1) -- (0,0);} \\
	\includegraphics[width=30mm]{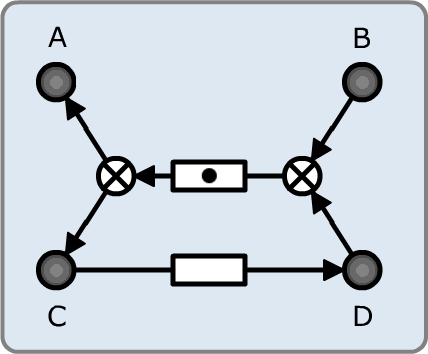} &
	\tikz{ \draw[->,thick,gray] (0,0) -- (1,0);} &
	\includegraphics[width=30mm]{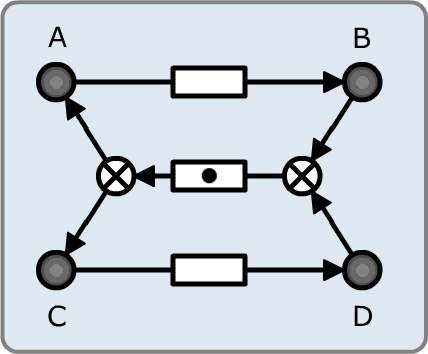} \\
\end{tabular}

	\end{center}
\caption{A pushout of Reo connectors.}
\label{fig:connector-pushout}
\end{figure}

\begin{example}
A pushout of Reo connectors in purely structural notation is depicted
in Fig.~\ref{fig:connector-pushout}. In this notation, nodes (which correspond
to port names) are depicted as filled circles. The port automata semantics for
the different channel types are given in Fig.~\ref{fig:primitives}.
So-called $FIFO$ channels are asynchronous channels with a buffer of size one.
They are represented as arrows with a rectangle in the middle.
There are in fact two versions of this channel type: with and
without a token, respectively called $FullFIFO$ and $EmptyFIFO$. Circles with a
cross denote two dual primitives: the $Router$ (left of the $FullFIFO$)
and the $Merger$ (right of it). Both have in total three ports and the same
semantics (cf.~Fig.~\ref{fig:primitives} for their port automata semantics).

Note that we abused notation in this example, in the sense that there are two hidden
nodes between the $Router$, the $FullFIFO$ and the $Merger$.
The resulting connector on the bottom right is the initial example from
Fig.~\ref{fig:examples}. It consists in total of five primitives (three
$FIFO$s, one $Merger$ and one $Router$) and six nodes ($A$-$D$, plus two hidden ones).
%The morphisms in this example do not include a refinement of primitives.
\end{example}

\section{Compositional Semantics}
\label{sec:semantics}

In this section we show how to compute the port automaton for 
a connector using its primitives' semantics. We extend 
this mapping to a functor and show compositionality.

\begin{remark}
In the following definitions we use the fact that the product
is associative and commutative, i.e. there exist
natural isomorphisms $(A \times B) \times C \cong A \times (B \times C)$ 
and $A \times B \cong B \times A$.
\end{remark}
\begin{definition}[Connector semantics] 
\label{def:connector-sem}
Given a connector $\C=(\A,\N)$ with $\A=\{A_1,\ldots, A_n\}$ and 
$A_j = (Q_j,N_j, \to_j, i_j)$ for $j\in \{1,\ldots,n\}$. We define
\[
\Sem(\C)=(Q_1 \times \ldots \times Q_n,\N,\to,\langle i_1, \ldots,i_n
\rangle)
\] 
where $\to$ is given by:
\begin{equation}
\frac{
\forall j,k \in \{1,\ldots,n\}: \; 
 q_j \overset{S_j}{\longrightarrow}_j p_j, \quad 
 q_k \overset{S_k}{\longrightarrow}_k p_k, \quad
 S_j \cap N_k = S_k \cap N_j
}
{
 \xymatrix{
	\langle q_1,\ldots,q_n \rangle 
	\ar[rr]^{S_1 \cup \ldots \cup S_n} &&
	\langle p_1,\ldots,p_n \rangle }
}
\label{rule:connector-sem}
\end{equation}
\end{definition}
Note that the nodes of the connector become the port names
of the resulting port automaton. The definition further
implies that all actions on a (shared) node are synchronized. 
This corresponds to so-called Hoare-style synchronizations, as
opposed to Milner-style synchronizations where exactly one
input end is synchronized with one output end.
We extend now the given connector semantics to a functor
$\Sem: \Connector \to \PA^{op}$.

\begin{theorem}[Semantics functor]
\label{def:functor}
Let $f=(f_\A, f_\N): \C_1 \to \C_2$ 
a connector morphism, 
$\Sem(\C_1) = (\Q_1,\N_1,\to_1,\iota_1)$ and 
$\Sem(\C_2) = (\Q_2,\N_2,\to_2,\iota_2)$ 
where $\C_1 = (\A_1,\N_1)$
with $\A_1 = (A_1,\ldots,A_n)$
and $f_{A_j} = (f_{Q_j},f_{N_j}) : f_\A(A_j) \to A_j$.
Let
\[
f_\Q = f_{Q_1} \times \ldots \times f_{Q_n} : Q \to \Q_1,
\]
then there also exists a projection $\pi_Q: \Q_2 \to Q$.
Defining $\Sem(f) : \Sem(\C_2) \to \Sem(\C_1)$
as $\Sem(f) = (f_\Q \circ \pi_Q,f_\N)$
gives rise to a contravariant functor 
$\Sem: \Connector \to \PA^{op}$.
\end{theorem}
\begin{proof}
$Q$ is the product of the state sets of those primitives
in $\C_2$ that are in the image of $f_\A$, and
$\pi_Q$ the projection to the product of the state sets 
of these reached primitives.
The states map of $\Sem(f)$ preserves the transitions
of $\Sem(\C_2)$, since both $\pi_Q$ and $f_\Q$ do.
Hence, $\Sem(f)$ is a valid port automata morphism.
For showing that $\Sem$ is a functor we
observe that composition is preserved: $\Sem(g \circ
f) = \Sem((g_\A,g_\N) \circ (f_\A,f_\N)) =
\Sem(g_\A \circ f_\A, g_\N \circ f_\N) = 
(f_\Q \circ g_\Q, f_\N \circ g_\N) =
(f_\Q \circ f_\N) \circ (g_\Q,g_\N) = 
\Sem(f) \circ \Sem(g)$, and analogously for the identity.
\end{proof}
This result in particular shows that a (structural) morphism
of connectors corresponds to an inverse simulation on the
semantical level. 
We now phrase our main result, i.e., the compositionality
of the port automata semantics for connectors.
\begin{theorem}[Compositionality of semantics]
\label{thm:compositionality}
The functor $\Sem$ maps pushouts of connectors to pullbacks of port automata,
i.e.
\[
\Sem(\C_1 +_{\C_0} \C_2) = \Sem(\C_1) \times_{\Sem(\C_0)} \Sem(\C_2).
\]
\end{theorem}
\begin{proof}
Both in {\Connector} and {\PA} the port name sets
are composed using pushouts in \Set, and $\Sem$ 
preserves these sets.  Hence, the port names are 
correctly mapped.
The primitives in the connector pushout
$\C_3 = \C_1 +_{\C_0} \C_2$ are either of the form
$A_1 \times_{A_0} A_2$ (case (\ref{equ:connector-pushout-1}))
or $A_j = A_j \times_\one \one$ with $j\in \{1,2\}$
(case (\ref{equ:connector-pushout-2})). The primitives' state
sets are of the same form, i.e. they can be all 
written as pullbacks. {\Sem} sends these 
state sets to their product. Now, since 
\[(X \times_Y Z) \times (X' \times_{Y'} Z') =
(X\times X') \times_{(Y\times Y')} (Z \times Z')\]
the state set of the resulting automaton $A_3 = \Sem(\C_3)$
is of the form $Q_1\times_{Q_0}Q_2$ where $Q_1$ is the state
set of $\Sem(\C_1)$, and $Q_2$ of $\Sem(\C_2)$ and
$Q_0$ of $\Sem(\C_0)$.  Hence, the (initial) states 
are also correctly mapped.
Moreover, the transition structure is preserved, 
since (\ref{rule:connector-sem}) implies (\ref{equ:caut-po})
and the port name sets on transitions are in both cases
composed by taking their union.
\end{proof}

\section{Applications}
\label{sec:encodings}

In this section, we show how Reo connectors and Petri nets
are modeled by our notion of connectors. This enables us
also to do a direct comparison of the two models.

\subsection{Modeling Reo connectors}

Reo connectors are directly modeled by our notion of connectors.
% Reo is a channel-based coordination language.
% An interesting fact about Reo is that channels
% are not necessarily directed, i.e.  
% there are also so-called \emph{drain} and \emph{spout}
% channels, which have respectively two source / 
% target ends. A standard semantics for Reo
% are so-called \emph{constraint automata}~\cite{BSAR06}.
% These automata model synchronization and further
% dataflow constraints on channels and connectors.
% Port automata are 
% an abstraction of constraint automata, which forget
% about the dataflow constraints. Due to this fact, 
% all primitives with constraint automata semantics can
% be also modeled using port automata. 
The primitives used in this paper are summarized in 
Fig.~\ref{fig:primitives}. It includes in particular 
the channel types $Sync$, $EmptyFIFO$ and $FullFIFO$.
Note also that all primitives
explicitly include $\tau$-steps to 
allow interleavings, i.e. other parts of the connector
can fire independently without (observable) activity of such 
a primitive. 

\begin{figure}[t]
\centering
\input{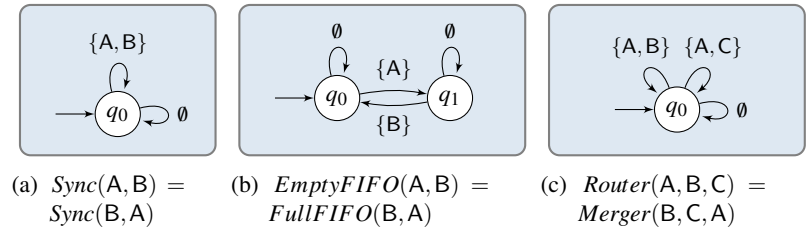}
\caption{Port automata for some Reo primitives.}
\label{fig:primitives}
\end{figure}

While channels are user-defined
entities, Reo defines a fixed semantics for nodes.
A node in Reo merges input from all target ends
ends and replicates it to all source ends. This
can be seen as a $1:n$ synchronization, 
as opposed the Hoare-style synchronizations
in our framework, where basically all
coinciding channel ends (no matter if source or target) 
are synchronized.
As a consequence we have to model the merging 
explicitly using a primitive. The $Merger$,
denoted by a circle with a cross in Fig.~\ref{fig:examples},
is used for this purpose. It has two source and one target end
and is therefore not a channel. We also define the dual
of this primitive, called $Router$. It has the same semantics
but its ends are inverse. Note again, that in the example
of Fig.~\ref{fig:primitives} and~\ref{fig:connector-pushout}
there are two hidden nodes between the $Router$, the $FullFIFO$
and the $Merger$, which are not relevant here.

The pushout diagram in Fig.~\ref{fig:connector-pushout} shows
a gluing of two Reo connectors along a common subconnector.
Note that this gluing is of a purely structural nature, 
although -- in principle -- it could also include a refinement of
primitives, i.e., if one of the primitive simulations is not 
an isomorphism. The port automata corresponding to the 
connectors in this example are depicted in
Fig.~\ref{fig:port-automata-pullback}. As we have shown 
in Theorem~\ref{thm:compositionality}, they form a pullback.
Note again, that we omitted unreachable states in the result.

\subsection{Modeling Petri nets}

Petri nets can also be modeled directly with our connector notion. As illustrated in the
motivating example in Fig.~\ref{fig:examples}, transitions in a Petri net
should be interpreted as nodes in this setting. Hence, places
become the primitives in the connector model. They are basically unbounded buffers
without ordering constraints (as opposed to the \emph{FIFO} channel in Reo).

Formally, the port automaton
$A_p=(Q,N,\to,i)$ for a place $p$ with $in(p)$ and $out(p)$ respectively 
the sets of incoming and outgoing transitions\footnote{This can be seen
as the dual of the usual notion of \emph{pre}- and \emph{post}-sets of Petri net transitions.}
of $p$ is defined in the following way:
\begin{itemize}
  \item $Q$ is the set of all markings of $p$, e.g. the natural numbers, or a
  finite set for places with capacities.
  \item $m \overset{T}{\longrightarrow} m'$ whenever a concurrent firing of
  the transitions $T \subseteq in(p) \cup out(p)$ turns the marking $m$ into $m'$. 
  \item $i\in Q$ the initial marking of $p$.
\end{itemize}
This encoding works because the transitions in a Petri net also do a basic
Hoare-style synchronization. Without giving a proof, we claim that the port
automaton $\Sem(N)$ of a Petri net $N$ correctly models its behavior, in the
sense that it has the set of all possible markings of the net as states and
transitions that correspond to a concurrent firing of net transitions.

Our notion of connector morphisms requires that the ports of primitives are preserved.
Since we interprete places as primitives (and transitions as nodes) our connectors
correspond to the following Petri net model:
\vspace{-0.2cm}
\[
\xymatrix{
N=P \ar@<0.5ex>[r]^{in} \ar@<-0.5ex>[r]_{out} & T^{\oplus}.
}
\]
Note that in the literature~(see e.g.~\cite{PER01}) one often finds
a similar but different model of Petri nets, where instead of the maps $in,out$ 
functions $pre,post: T \to P^{\oplus}$ are used. However, this only modifies the
notion of net morphism, but not the Petri net model itself. A comparison of the
two types of Petri net morphisms is out of the scope of this paper.

\subsection{Comparing Reo connectors and Petri nets}

As made evident in this paper, one can compare the basic version of Reo nodes, which
does only a primitive synchronization, with the transition concept in Petri nets.
On the other hand, primitives in our framework correspond to the places in a Petri net
and the channels, mergers, routers etc. in a Reo connector. Reo is more expressive
in the sense that it allows unbuffered  primitives, such as synchronous communication 
channels. In Petri nets, the primitives, i.e., the places are always buffered. 
From our point of view, this is the most important difference between Reo connectors 
and Petri nets: while in Petri nets synchronizations happen always locally at transitions,
in Reo synchronous primitives can be used to \emph{propagate} synchrony through the 
connector. Other features of Reo, such as context-dependency and priority go
beyond the focus of this paper.

\subsection{Modeling Reconfigurations}

Graph rewriting techniques, such as the double pushout (DPO) approach~\cite{DPO},
are a powerful tool for modeling rule-based reconfigurations. 
As a motivating example we return to the
pushout of Reo connectors in Fig.~\ref{fig:connector-pushout}. This diagram can
be interpreted as a reconfiguration in the following sense. The upper two
connectors together with the morphism between them is interpreted as a
(structural) reconfiguration rule. An application of this rule creates a
new $EmptyFIFO$ between the nodes $\mathsf{A}$ and $\mathsf{B}$.
The bottom two connectors can be regarded as an application of this rule.
In the bottom left is the connector before, and in the bottom right after the
rule application.

In our approach we can now perform this reconfiguration directly on the
corresponding automata. This becomes particularly interesting when
executing such a Reo connector as a state machine and reconfiguring it at
runtime. A prototypic implementation of this approach exists already.
A typical question in such scenarios is what the state of the connector
after a reconfiguration is, and whether it is actually valid. 
We can make this clear in the corresponding automata pullback in
Fig.~\ref{fig:port-automata-pullback}. For instance, if the connector 
before the reconfiguration is in state $r_1$, we can see that in the 
connector after the reconfiguration the state $s_1$ is mapped to $r_1$
by the constructed morphism. Thus, we can use this morphism to identify
the state after a reconfiguration step. However, we can see also in this example
that $r'_2$ has a preimage in the target automaton that is an unreachable state.
This indicates that a reconfiguration in this state produces an invalid system
state. In our example first $\mathsf{B}$ and then $\mathsf{C}$ fired, and then
the reconfiguration was performed. The problem here is essentially that at this point
there are two tokens in the connector.

\section{Conclusions and Future Work}
\label{sec:discussion}

We have presented an integrated structural and behavioral
model of connectors and showed compositionality with
respect to gluing constructions. We have then shown how Reo 
connectors and Petri nets can be modeled in this framework.

As future work, we would like to consider the traditional model 
of simulations for port automata morphisms, i.e. instead
of functions we want to use a notion of upward-closed relations 
for relating the states. With this change, the model will cover 
a wider class of connector morphisms. Moreover, we are 
interested in further properties of the semantical functor.

\section{Related work}
\label{sec:related-work}

Padberg et al. provide compositional semantics of Petri nets
in~\cite{PER01}. Their results are based on 
\emph{pre}/\emph{post}-net morphisms and a marking graph semantics,
and they cover a wider class of Petri nets. Moreover, the 
authors show preservation of general colimits, as opposed 
to our work were we consider only pushouts.

A wide range of automata semantics for Petri nets exist. 
For instance, Droste 
and Shortt consider so-called \emph{automata
with concurrency relations} in~\cite{DS02}, which are more 
restrictive than port automata. Essentially, a concurrent 
firing of two net transitions always implies the existence 
of an interleaved execution of the two net-transitions 
(parallel independence).The authors show that there is a
coreflection between the category of Petri nets and automata 
with concurrency.
Pushouts or general colimits are not considered.
They further also use a non-standard notion of net morphisms.

A compositional automata semantics for Reo,
called \emph{constraint automata}, is
given by Baier et al. in~\cite{BSAR06}. 
Our port automata are an abstraction of 
constraint automata. The main difference
is the used notion of compositionality.
In~\cite{BSAR06}, with compositionality 
the authors mean that the semantics of a 
connector can be computed out of the semantics
of its constituent primitives. However, our
notion of compositionality really combines
the structural level with the semantical,
in the sense that we show how a gluing of
connectors corresponds to a join operation
of their behaviors. In particular, we generalize
the join operation of~\cite{BSAR06} by allowing
to join two automata along a common interface 
automaton.

\nocite{BCEH05}

\bibliographystyle{eptcs} % or whatever you prefer
\bibliography{ice02}

\begin{thebibliography}{10}
\providecommand{\bibitemstart}[1]{\bibitem{#1}}
\providecommand{\bibitemend}{}
\providecommand{\bibliographystart}{}
\providecommand{\bibliographyend}{}
\providecommand{\url}[1]{\texttt{#1}}
\providecommand{\urlprefix}{Available at }
\providecommand{\bibinfo}[2]{#2}
\bibliographystart

\bibitemstart{Arbab04}
\bibinfo{author}{F.~Arbab} (\bibinfo{year}{2004}): \emph{\bibinfo{title}{Reo:
  {A} {C}hannel-based {C}oordination {M}odel for {C}omponent {C}omposition}}.
\newblock {\sl \bibinfo{journal}{Mathematical Structures in Computer Science}}
  \bibinfo{volume}{14}, pp. \bibinfo{pages}{329--366}.
\bibitemend

\bibitemstart{Tiles}
\bibinfo{author}{F.~Arbab}, \bibinfo{author}{R.~Bruni},
  \bibinfo{author}{D.~Clarke}, \bibinfo{author}{I.~Lanese} \&
  \bibinfo{author}{U.~Montanari} (\bibinfo{year}{2009}):
  \emph{\bibinfo{title}{Tiles for {R}eo}}.
\newblock In: {\sl \bibinfo{booktitle}{Recent {T}rends in {A}lgebraic
  {D}evelopment {T}echniques (WADT'09)}}, {\sl \bibinfo{series}{Lecture Notes
  in Computer Science}} \bibinfo{volume}{5486}. \bibinfo{publisher}{Springer},
  pp. \bibinfo{pages}{37--55}.
\bibitemend

\bibitemstart{BSAR06}
\bibinfo{author}{C.~Baier}, \bibinfo{author}{M.~Sirjani},
  \bibinfo{author}{F.~Arbab} \& \bibinfo{author}{J.~Rutten}
  (\bibinfo{year}{2006}): \emph{\bibinfo{title}{Modeling {C}omponent
  {C}onnectors in {R}eo by {C}onstraint {A}utomata}}.
\newblock {\sl \bibinfo{journal}{Science of Computer Programming}}
  \bibinfo{volume}{61}(\bibinfo{number}{2}), pp. \bibinfo{pages}{75--113}.
\bibitemend

\bibitemstart{BCEH05}
\bibinfo{author}{P.~Baldan}, \bibinfo{author}{A.~Corradini},
  \bibinfo{author}{H.~Ehrig} \& \bibinfo{author}{R.~Heckel}
  (\bibinfo{year}{2005}): \emph{\bibinfo{title}{Compositional {S}emantics for
  {O}pen {P}etri {N}ets based on {D}eterministic {P}rocesses}}.
\newblock {\sl \bibinfo{journal}{Mathematical Structures in Computer Science}}
  \bibinfo{volume}{15}, pp. \bibinfo{pages}{1--35}.
\bibitemend

\bibitemstart{BCS09}
\bibinfo{author}{M.~Bonsangue}, \bibinfo{author}{D.~Clarke} \&
  \bibinfo{author}{A.~Silva} (\bibinfo{year}{2009}):
  \emph{\bibinfo{title}{{A}utomata for {C}ontext-dependent {C}onnectors}}.
\newblock In: {\sl \bibinfo{booktitle}{Proceedings of 11th {I}nternational
  {C}onference on {C}oordination {M}odels and {L}anguages, {C}oordination'09}},
  {\sl \bibinfo{series}{Lecture Notes in Computer Science}}
  \bibinfo{volume}{5521}. \bibinfo{publisher}{Springer}, pp.
  \bibinfo{pages}{184--203}.
\bibitemend

\bibitemstart{CCA07}
\bibinfo{author}{D.~Clarke}, \bibinfo{author}{D.~Costa} \&
  \bibinfo{author}{F.~Arbab} (\bibinfo{year}{2007}):
  \emph{\bibinfo{title}{Connector {C}olouring {I}: {S}ynchronisation and
  {C}ontext {D}ependency}}.
\newblock {\sl \bibinfo{journal}{Science of Computer Programming}}
  \bibinfo{volume}{66}(\bibinfo{number}{3}), pp. \bibinfo{pages}{205--225}.
\bibitemend

\bibitemstart{Clarke07}
\bibinfo{author}{Dave Clarke} (\bibinfo{year}{2007}):
  \emph{\bibinfo{title}{Coordination: Reo, Nets, and Logic}}.
\newblock In: {\sl \bibinfo{booktitle}{Formal Methods for Components and
  Objects (FMCO)}}, {\sl \bibinfo{series}{Lecture Notes in Computer Science}}
  \bibinfo{volume}{5382}. pp. \bibinfo{pages}{226--256}.
\bibitemend

\bibitemstart{DPO}
\bibinfo{author}{A.~Corradini}, \bibinfo{author}{U.~Montanari},
  \bibinfo{author}{F.~Rossi}, \bibinfo{author}{H.~Ehrig},
  \bibinfo{author}{R.~Heckel} \& \bibinfo{author}{M.~L\"owe}
  (\bibinfo{year}{1997}): \emph{\bibinfo{title}{Handbook of Graph Grammars and
  Computing by Graph Transformation}}, chapter \bibinfo{chapter}{Algebraic
  {A}pproaches to {G}raph {T}ransformation~{I}: {B}asic {C}oncepts and {D}ouble
  {P}ushout {A}pproach}, pp. \bibinfo{pages}{163--245}.
\newblock \bibinfo{publisher}{World Scientific}.
\bibitemend

\bibitemstart{DS02}
\bibinfo{author}{M.~Droste} \& \bibinfo{author}{R.~M. Shortt}
  (\bibinfo{year}{2002}): \emph{\bibinfo{title}{From {P}etri {N}ets to
  {A}utomata with {C}oncurrency.}}
\newblock {\sl \bibinfo{journal}{Applied Categorical Structures}}
  \bibinfo{volume}{10}(\bibinfo{number}{2}), pp. \bibinfo{pages}{173--191}.
\bibitemend

\bibitemstart{EEPT06}
\bibinfo{author}{H.~Ehrig}, \bibinfo{author}{K.~Ehrig},
  \bibinfo{author}{U.~Prange} \& \bibinfo{author}{G.~Taentzer}
  (\bibinfo{year}{2006}): \emph{\bibinfo{title}{{F}undamentals of {A}lgebraic
  {G}raph {T}ransformation}}.
\newblock EATCS Monographs in Theoretical Computer Science.
  \bibinfo{publisher}{Springer}.
\bibitemend

\bibitemstart{SABB06}
\bibinfo{author}{J.~Guillen-Scholten}, \bibinfo{author}{F.~Arbab},
  \bibinfo{author}{F.~de~Boer} \& \bibinfo{author}{M.~Bonsangue}
  (\bibinfo{year}{2006}): \emph{\bibinfo{title}{A {C}omponent {C}oordination
  {M}odel {B}ased on {M}obile {C}hannels}}.
\newblock {\sl \bibinfo{journal}{Fundamenta Informaticae}}
  \bibinfo{volume}{73}(\bibinfo{number}{4}), pp. \bibinfo{pages}{561--582}.
\bibitemend

\bibitemstart{KAV09}
\bibinfo{author}{C.~Koehler}, \bibinfo{author}{F.~Arbab} \&
  \bibinfo{author}{E.~de~Vink} (\bibinfo{year}{2009}):
  \emph{\bibinfo{title}{Reconfiguring {Distributed} {R}eo {C}onnectors}}.
\newblock In: {\sl \bibinfo{booktitle}{Recent {T}rends in {A}lgebraic
  {D}evelopment {T}echniques (WADT'09)}}, {\sl \bibinfo{series}{Lecture Notes
  in Computer Science}} \bibinfo{volume}{5486}. \bibinfo{publisher}{Springer},
  pp. \bibinfo{pages}{221--235}.
\bibitemend

\bibitemstart{KC09}
\bibinfo{author}{C.~Koehler} \& \bibinfo{author}{D.~Clarke}
  (\bibinfo{year}{2009}): \emph{\bibinfo{title}{Decomposing {P}ort
  {A}utomata}}.
\newblock In: {\sl \bibinfo{booktitle}{Proceedings of 24th {A}nnual {ACM}
  {S}ymposium on {A}pplied {C}omputing, {SAC}'09}}. \bibinfo{publisher}{ACM}.
\bibitemend

\bibitemstart{KMLA10}
\bibinfo{author}{C.~Krause}, \bibinfo{author}{Z.~Maraikar},
  \bibinfo{author}{A.~Lazovik} \& \bibinfo{author}{F.~Arbab}
  (\bibinfo{year}{2010}): \emph{\bibinfo{title}{{M}odeling {D}ynamic
  {R}econfigurations in {R}eo using {H}igh-{L}evel {R}eplacement {S}ystems}}.
\newblock {\sl \bibinfo{journal}{Science of Computer Programming (to appear)}}
  .
\bibitemend

\bibitemstart{PER01}
\bibinfo{author}{J.~Padberg}, \bibinfo{author}{H.~Ehrig} \&
  \bibinfo{author}{G.~Rozenberg} (\bibinfo{year}{2001}):
  \emph{\bibinfo{title}{Behavior and {R}ealization {C}onstruction for {P}etri
  {N}ets {B}ased on {F}ree {M}onoid and {P}ower {S}et {G}raphs}}.
\newblock In: {\sl \bibinfo{booktitle}{Unifying Petri Nets, Advances in Petri
  Nets}}. \bibinfo{publisher}{Springer-Verlag}, \bibinfo{address}{London, UK},
  pp. \bibinfo{pages}{230--249}.
\bibitemend

\bibliographyend
\end{thebibliography}

\appendix

\section{Proofs}
\label{sec:proofs}

\begin{proof}[Theorem~\ref{thm:pullbacks}]
Let $j\in \{1,2\}$. The pullback morphisms $g_j=(g_{j,Q},g_{j,N})$ consist of
the projections $g_{j,Q}: Q_3\to Q_j$ and the injections $g_{j,N}: N_j \to
N_3$. We will denote both of them with $g_j$ if there is no confusion. We have
to show first that $g_1$ and $g_2$ are in fact \PA-morphisms.
Condition~(\ref{equ:aut-mor}) in the morphism definition reads for $g_j$:  
\[ 
	g_j(N_j) \cap \bigl(g_1(S_1) \cup g_2(S_2) \bigr) = g_j(S_j) 
\]
We show this here only for $j=1$, since the other case is analogously:
\vspace{3mm}\\
$g_1(N_1) \cap \bigl(g_1(S_1) \cup g_2(S_2)\bigr)$
\vspace{-3mm}
\setlength{\extrarowheight}{2mm}
\[
\begin{array}{rclr}
			 & = & \bigl(g_1(N_1) \cap g_1(S_1)\bigr) \cup \bigl(g_1(N_1) \cap g_2(S_2)\bigr) \\
			 & = & g_1(S_1) \cup \bigl(g_1(N_1) \cap g_2(S_2)\bigr) 	& \hspace{1cm} \mbox{since } S_1\subseteq N_1 \\
			 & = & g_1(S_1) \cup \bigl(g_2(N_2) \cap g_1(S_1)\bigr) 	& \mbox{by (\ref{equ:caut-po})} \\
			 & = & g_1(S_1)
\end{array}
\]
Now, the arrow $h: A_3 \to X$ exists and is unique due to the componentwise construction in \Set. What is left
to show is that $h$ is also a valid \PA-morphism.
We know for all $q \overset{N}{\longrightarrow} p$ in $X$
there exist transitions $h_j(q) \overset{S_j}{\longrightarrow}_j h_j(p)$ in $A_j$ with
\begin{equation}
	\label{equ:caut-po-mors}
	h_j(N_j) \cap N = h_j(S_j)
\end{equation}
since the $h_j$ are by assumption valid morphisms. Moreover we know that $h$ maps
a state $q\in Q_X$ in the automaton $X$ to the state
$h(q) = \langle h_1(q),h_2(q) \rangle$
in the pushout object $A_3$. Now we have to show that there exists a transition 
$h(q) \overset{S_3}{\longrightarrow}_3 h(p)$ in $A_3$ with 
\[ h(N_3) \cap N = h(S_3). \]
We construct $S_3$ in the following way:
\[
\begin{array}{rclr}
	 h\bigl(N_3\bigr) \cap N & = & h\bigl(g_1(N_1) \cup g_2(N_2)\bigr) \cap N 	& \hspace{1cm} \mbox{pushout in \Set} \\
			& = & \bigl(h\circ g_1(N_1) \cup h\circ g_2(N_2)\bigr) \cap N \\
			& = & \bigl(h_1(N_1) \cup h_2(N_2)\bigr) \cap N 					& \mbox{since } h\circ g_j = h_j \\
			& = & \bigl(h_1(N_1) \cap N\bigr) \cup \bigl(h_2(N_2) \cap N\bigr)  \\
			& = & h_1(S_1) \cup h_2(S_2) 										& \mbox{by (\ref{equ:caut-po-mors})} \\
			& = & h\circ g_1(S_1) \cup h\circ g_2(S_2) 							& \mbox{since } h\circ g_j = h_j \\
		 	& = & h\bigl(g_1(S_1) \cup g_2(S_2)\bigr) \\
\end{array}
\]
We have constructed $S_3 = g_1(S_1) \cup g_2(S_2)$ and it fulfills the required property.
The last step is to show that this transition in fact exists in $A_3$, which means that 
(\ref{equ:caut-po}) holds. Recall that $g_1$ and $g_2$ are valid morphisms:
\[
	g_1(N_1) \cap S_3 = g_1(S_1) 
	\quad \mbox{and} \quad	 
	g_2(N_2) \cap S_3 = g_2(S_2).
\]
We can follow that:
\begin{itemize}
\item $\bigl(g_1(N_1) \cap S_3\bigr) \cup g_2(S_2) = g_1(S_1) \cup g_2(S_2)$
\item $\bigl(g_2(N_2) \cap S_3\bigr) \cup g_1(S_1) = g_2(S_2) \cup g_1(S_1)$
\end{itemize}
and unify both equations:
\[
\begin{array}{rclr}
	(g_1(N_1) \cap S_3) \cup g_2(S_2) 											& = & (g_2(N_2) \cap S_3) \cup g_1(S_1) & \Leftrightarrow \\
	\Bigr(g_1(N_1) \cap \bigl(g_1(S_1) \cup g_2(S_2)\bigr)\Bigr) \cup g_2(S_2)	& = & \Bigl(g_2(N_2) \cap \bigl(g_1(S_1) \cup g_2(S_2)\bigr)\Bigr) \cup g_1(S_1) & \Leftrightarrow \\
	\bigl(g_1(N_1) \cup g_2(S_2)\bigr) \cap \bigl(g_1(S_1) \cup g_2(S_2)\bigr)	& = & \bigl(g_2(N_2)\cup g_1(S_1)\bigr) \cap \bigl(g_1(S_1) \cup g_2(S_2)\bigr) & \Leftrightarrow \\
	\bigl(g_1(N_1) \cup g_1(S_1)\bigr) \cap g_2(S_2)							& = & \bigl(g_2(N_2)\cup g_2(S_2)\bigr) \cap g_1(S_1) & \Leftrightarrow \\
	g_1(N_1) \cap g_2(S_2)														& = & g_2(N_2) \cap g_1(S_1) \\
\end{array}
\]
and we have shown (\ref{equ:caut-po}).
\end{proof}

\end{document}